\documentclass[10pt, conference]{IEEEtran}
\IEEEoverridecommandlockouts 
\usepackage{cite}
\usepackage{amsmath,amssymb,amsfonts,amsthm}
\usepackage{mathtools}
\usepackage{bm}
\usepackage{algorithmic}
\usepackage{makecell}
\usepackage{graphicx}
\usepackage{textcomp}
\usepackage{enumitem}
\usepackage{comment}
\usepackage{xcolor}
\usepackage{booktabs}
\newtheorem{theorem}{Theorem}[section]
\newtheorem{lemma}[theorem]{Lemma}

\newtheorem{definition}[theorem]{Definition}

\newtheorem{remark}[theorem]{Remark}

\newcommand{\alphaDL}{\boldsymbol{\alpha}^{\mathrm{DL}}}
\newcommand{\alphaUL}{\boldsymbol{\alpha}^{\mathrm{UL}}}
\newcommand{\bfAlpha}{\boldsymbol{\alpha}}
\newcommand{\UDL}{\mathcal{U}_{\mathrm{DL}}}
\newcommand{\UUL}{\mathcal{U}_{\mathrm{UL}}}

\newcommand{\JDL}{J_{\mathrm{DL}}}
\newcommand{\JUL}{J_{\mathrm{UL}}}
\newcommand{\Cfib}{\hat{C}^{\mathrm{fib}}}
\newcommand{\Cntn}{\hat{C}^{\mathrm{ntn}}}
\newcommand{\bhat}{\hat{b}}
\newcommand{\Pot}{\Phi}

\begin{document}

\title{SLA-Aware Traffic Steering in Hybrid TN-NTN 5G Backhaul: A Potential Game Approach}

\author{\IEEEauthorblockN{Hojjat Navidan\IEEEauthorrefmark{1}, Delia Rico\IEEEauthorrefmark{2}, Mohammad Cheraghinia\IEEEauthorrefmark{1}, Ingrid Moerman\IEEEauthorrefmark{1}, Adnan Shahid\IEEEauthorrefmark{1}} \\
\IEEEauthorblockA{\IEEEauthorrefmark{1}IDLab, Department of Information Technology, Ghent University – imec, Belgium}
\IEEEauthorblockA{\IEEEauthorrefmark{2} ITIS Software, University of Malaga, Spain 
\thanks{This work was supported by the Horizon Europe MCSA Staff Exchanges 2021 program under grant agreement 101086218 (EVOLVE); the Smart Networks and Services Joint Undertaking (SNS JU) under the EU Horizon Europe research and innovation program under Grant Agreements No. 101139194 (6G-XCEL project) and No. 101096328 (6G-SANDBOX project); and the Spanish national project LearnFDT under grant agreement PID2022-142181OB-I00.}
}}

\maketitle

\begin{abstract}
The integration of Non-Terrestrial Networks (NTN) with Terrestrial Networks (TN) is a key enabler for resilient 5G-Advanced and future 6G backhaul infrastructures. However, managing traffic across these highly asymmetric links remains a significant routing challenge, as systems must support heterogeneous network slices with conflicting service-level agreements (SLAs) while selectively utilizing costly NTN resources. 
This paper presents a computationally lightweight SLA-aware traffic-steering framework for a hybrid TN-NTN backhaul that models the load-balancing problem as an exact potential game. This mathematical foundation inherently enables decentralized coordination between uplink and downlink load-balancing agents without control-message overhead. By formulating traffic steering as a coupled optimization problem, per-slice (or per-user group) traffic fractions are dynamically distributed across terrestrial and satellite paths based on utility functions that capture throughput, latency, packet loss, and SLA penalties. The resulting game admits a pure Nash equilibrium, ensuring stable and predictable traffic adaptation under non-stationary load conditions. The framework is evaluated on a geographically distributed 5G testbed, 
using bidirectional traffic generated for five representative slices. Experimental results show that the proposed controller significantly outperforms heuristic and conventional baselines, reducing SLA violations to 1.7\% for V2X and 0.7\% for the emergency slice while completely eliminating them for video, IoT, and best-effort traffic.  
\end{abstract}


\section{Introduction}
The evolution toward 5G-Advanced and emerging 6G network architectures relies heavily on the seamless integration of Non-Terrestrial Networks (NTN), such as Low Earth Orbit (LEO) satellite constellations, with traditional Terrestrial Networks (TN) \cite{azariEvolutionNonTerrestrialNetworks2022,3GPP}. This TN-NTN hybrid infrastructure is essential to ensure global coverage and network resilience during outages \cite{shanwongIntegrationNonterrestrialNetwork2026}. Augmenting fiber-optic backhauls with satellite links enables network operators to dynamically pool capacity to meet the increasing bandwidth demands of next-generation mobile applications \cite{diUltraDenseLEOIntegrating2019}.

Furthermore, the 5G network-slicing paradigm requires this unified backhaul infrastructure to support highly heterogeneous traffic classes with strictly enforced, yet conflicting, Service-Level Agreements (SLAs) \cite{zhangNetworkSlicingBased2017}. These SLA constraints span a diverse range of network slices, from latency-sensitive Vehicle-to-Everything (V2X) and emergency services, to massive Internet of Things (IoT) slices designed for delay-tolerant data bursts. Satisfying these diverse Quality of Service (QoS) requirements demands that the network edge possess the intelligence to seamlessly steer per-User Equipment (UE) traffic flows across available backhaul paths. This will ultimately maximize the overall network performance in terms of throughput, latency, and packet loss, ensuring a consistent user experience without starving any individual slice \cite{puglieseIntegratingTerrestrialNonterrestrial2024}.

However, managing traffic across a hybrid TN-NTN architecture poses complex routing challenges due to severe link asymmetry and economic limitations \cite{puglieseIntegratingTerrestrialNonterrestrial2024}. While terrestrial fiber links offer low latency and high capacity, NTN links inherently suffer from substantial propagation delays and highly variable channel conditions \cite{xuLinkStateAwareHybrid2022}. Importantly, using NTN resources typically results in significantly higher operational costs due to the inherent expense of satellite bandwidth and the complex infrastructure required for maintenance \cite{sawadBackhaul5GSystems2023}. 

Therefore, satellite resources should not be used indiscriminately for all traffic, but rather employed selectively when terrestrial backhaul becomes congested or when resilience is required \cite{shikderTrafficFlowSteering2022a}. This calls for an intelligent load-balancing and traffic-steering mechanism that dynamically accounts for SLA requirements while balancing the latency penalties and resource costs associated with satellite transport. Traditional static load-balancing strategies fail to adapt to the highly asymmetric and cost-sensitive dynamics of hybrid TN-NTN links, risking severe QoS degradation and SLA violations under fluctuating loads.

We propose a computationally lightweight, game-theoretic, proactive traffic steering mechanism that dynamically steers per-slice (or per-UE group) traffic across asymmetric TN and NTN backhaul paths.
By incentivizing latency-tolerant traffic to offload to the NTN link during heavy load events, the proposed system proactively reserves the terrestrial path for mission-critical slices.

The core contributions of this paper are fourfold:
\begin{itemize}
    \item  We formulate the per-slice traffic steering problem as an exact potential game, mapping competing SLAs into utility functions governed by a shared link-capacity penalty.
    \item  We design a proactive, decentralized routing controller to solve the game. This computationally lightweight approach anticipates congestion and mathematically guarantees rapid convergence to a pure Nash equilibrium.
    \item We empirically validate the proposed framework on a live, geographically distributed 5G testbed, utilizing commercial Starlink NTN and terrestrial fiber links.
    \item We demonstrate that our approach outperforms static and reactive baselines,  reducing SLA violations for priority slices under non-stationary traffic.
\end{itemize}

\section{Related Work}

Integration of satellites into terrestrial backhauls has become a major research direction for beyond-5G and 6G systems 
\cite{azariEvolutionNonTerrestrialNetworks2022}.
Recent surveys highlight that hybrid TN-NTN architectures introduce substantial challenges in routing, resource allocation, and QoS assurance due to heterogeneous propagation delays, capacity, and channel dynamics \cite{tirmiziHybridSatelliteTerrestrial2022}.

In multi-RAT and 5G environments, distributed and centralized load-balancing schemes have been proposed to address QoS requirements and uneven traffic distribution 
\cite{seyoumDistributedLoadBalancing2023}. Similarly, 5G access traffic steering, switching, and splitting mechanisms split flows across heterogeneous access points based on utility and service requirements \cite{baMultiserviceBasedTrafficScheduling2022}. While these approaches demonstrate the value of adaptive traffic steering, they are typically developed for multi-access terrestrial settings and do not explicitly address the challenges of TN-NTN backhaul.

Focusing on the satellite integration, Shikder et al. \cite{shikderTrafficFlowSteering2022a} proposed a traffic-flow steering algorithm that classifies flows by delay tolerance and redirects suitable traffic to the satellite path under terrestrial congestion. 
More generally, recent hybrid TN-NTN studies emphasize dynamic path selection in response to changing link conditions~\cite{tirmiziHybridSatelliteTerrestrial2022}. However, most of these approaches rely on heuristic rules or focus on generic throughput or congestion objectives without explicitly encoding per-slice SLA penalties.

Proactive traffic steering in O-RAN has been shown to reduce queuing delay relative to reactive methods, and network-aided intelligent frameworks optimize flow splitting using utility-based formulations~\cite{tamimIntelligentORANTraffic2023}. Traffic classification for traffic steering using Deep Packet Inspection (DPI) or machine learning introduces additional processing overhead at the network edge~\cite{azabNetworkTrafficClassification2024}. Our solution avoids these bottlenecks entirely by performing SLA-aware traffic allocation at the IP transport layer using only aggregated utility metrics. Slice identification is performed by lightweight parsing of the encapsulated packet headers, which is less computationally demanding than full payload DPI.

Despite these advances, existing studies lack a unified optimization framework for hybrid TN-NTN backhaul. They neither enforce heterogeneous SLA constraints across slices nor explicitly model the coupled interaction between uplink and downlink traffic competing for shared resources. In contrast, our approach embeds SLA penalties directly into per-slice utility functions and solves the problem via a provably convergent exact potential game, guaranteeing stable, decentralized adaptation under non-stationary loads.

\section{System Model and Game Theory Framework}
\label{sec:model}
We consider a hybrid TN–NTN backhaul architecture as illustrated in Figure \ref{fig:archi}. Two independent traffic-steering agents (downlink and uplink) operate at opposite ends of the backhaul. Each agent dynamically allocates traffic fractions across a low-latency terrestrial fiber link and a high-latency, capacity-constrained NTN path. The agents make decisions locally using only their own observations and interface-level telemetry without exchanging any control messages. 

\begin{figure}[h]
    \centering
    \includegraphics[width=0.85\linewidth]{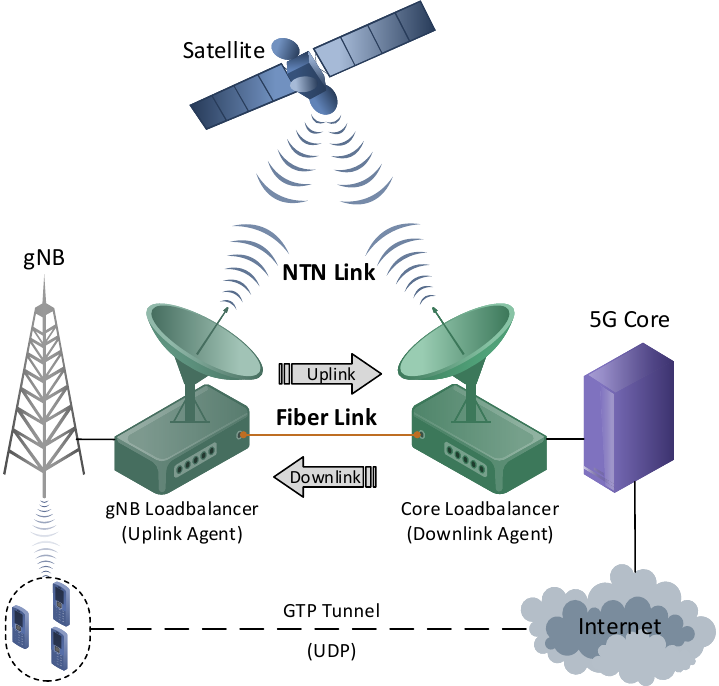}
    \caption{System overview of the considered hybrid backhaul, where traffic is dynamically steered between terrestrial fiber and NTN satellite paths.}
    \label{fig:archi}
\end{figure}

Fractional traffic splitting is enforced dynamically at the IP layer. The load balancers identify the corresponding slice type of individual UDP-based GPRS Tunneling Protocol (GTP) flows by inspecting packet headers within the encapsulated payload and apply routing rules to distribute traffic across the two backhaul paths. 

Each slice or UE-group $u$ is assigned a relative priority weight $\pi_u$, and a set of transport-layer SLA constraints, namely maximum Round-Trip Time (RTT), jitter, and packet loss. While 5G SLAs are traditionally defined end-to-end, we focus exclusively on the backhaul transport layer. The RAN segment is treated as a transparent conduit, and the controller optimizes only routing over the TN and NTN links.


\subsection{Game-Theoretic Framework}
\label{sec:game_framework}

Based on the two-agent definition, we define a noncooperative game as a triple 
\[
    \Gamma \;=\;
    \Bigl\{\,
        \mathcal{N},\;
        \bigl\{\mathcal{A}_i\bigr\}_{i \in \mathcal{N}},\;
        \bigl\{J_i\bigr\}_{i \in \mathcal{N}}
    \,\Bigr\},
\]
where $\mathcal{N} = \{\mathrm{DL},\, \mathrm{UL}\}$ is the set of
players, $\mathcal{A}_i$ is the strategy space of player $i$, and $J_i$ is the payoff
function of player $i$.

\paragraph{Players}
The downlink agent controls downstream traffic allocation for the $N$ slices in the set
$\UDL = \{1, \ldots, N\}$, and the uplink agent controls upstream allocation for the $M$ slices in
the set $\UUL = \{1, \ldots, M\}$.

\paragraph{Strategies}
Each player selects a vector of per-slice NTN allocation fractions. For the downlink,
$\alphaDL = (\alpha_1^{\mathrm{DL}}, \ldots, \alpha_N^{\mathrm{DL}}) \in [0,1]^N$,
where $\alpha_u^{\mathrm{DL}} \in [0,1]$ denotes the fraction of slice $u$'s downlink
traffic routed over the NTN link (and the remainder $1 - \alpha_u^{\mathrm{DL}}$
traverses the fiber link). Analogously,
$\alphaUL = (\alpha_1^{\mathrm{UL}}, \ldots, \alpha_M^{\mathrm{UL}}) \in [0,1]^M$
for the uplink agent. The joint strategy space is
$\mathcal{A} = [0,1]^N \times [0,1]^M$, a compact and convex set.

\paragraph{Payoff functions}
Each player's payoff function is a weighted sum of per-slice utilities minus a
shared capacity coupling penalty. Let $\alpha_u \in \{\alpha_u^{\mathrm{DL}},\alpha_u^{\mathrm{UL}} \}$ denote the respective player's allocation for slice $u$; its utility $U_u(\alpha_u)$ is defined as
\begin{multline}
    U_u(\alpha_u) \;=\;
        w_T^{c_u}\,T_u(\alpha_u)
        + w_L^{c_u}\,L_u(\alpha_u)
        + w_R^{c_u}\,R_u(\alpha_u) \\
        - w_P^{c_u}\,\lambda^{c_u}\,P_u(\alpha_u),
    \label{eq:per_ue_utility}
\end{multline}
where $w_T^{c_u}, w_L^{c_u}, w_R^{c_u}, w_P^{c_u} \geq 0$ are traffic-class-specific
weights, $\lambda^{c_u} \geq 0$ is a class severity multiplier, and the four component
terms are defined in Table~\ref{tab:utility_terms}.

Here $\Cfib$ and $\Cntn$ are the adaptive capacity estimates of the fiber and NTN links,
respectively; $d^{\mathrm{fib}}, d^{\mathrm{ntn}}$ are the measured round-trip latencies;
$d^{\max}_{c_u}$ is the SLA latency ceiling for traffic class $c_u$; $\ell^{\mathrm{fib}}$
and $\ell^{\mathrm{ntn}}$ are the fractional packet loss rates; $m_k(\alpha_u) =
\alpha_u\,\mathtt{metric}_{k,\mathrm{ntn}} + (1-\alpha_u)\,\mathtt{metric}_{k,\mathrm{fib}}$
is the blended SLA metric on dimension $k \in \{\mathrm{rtt, loss, jitter}\}$; and
$s_k^{c_u} = \mathrm{SLA}_k^{c_u}$ is the corresponding SLA threshold.

The capacity coupling penalty $C(\alphaDL, \alphaUL)$ captures the joint oversubscription cost shared by both agents. For each agent~$i \in \{\text{DL}, \text{UL}\}$, the load contributions to the NTN and fiber links are
\begin{align}
    D_{\mathrm{ntn}}^{i}(\alpha^i)
    &= \sum_{u \in \mathcal{U}_i} \alpha_u^{i}\,\bhat_u^{i}, &
    D_{\mathrm{fib}}^{i}(\alpha^i)
    &= \sum_{u \in \mathcal{U}_i} (1-\alpha_u^{i})\,\bhat_u^{i},
    \label{eq:own_load}
\end{align}
with $\bhat_u^{i} > 0$ denoting estimated peak throughput demands.
The aggregate link demands are
\begin{align}
    D_{\mathrm{ntn}}(\alphaDL, \alphaUL)
    &= D_{\mathrm{ntn}}^{DL}(\alphaDL) + D_{\mathrm{ntn}}^{UL}(\alphaUL),
    \label{eq:D_ntn} \\
    D_{\mathrm{fib}}(\alphaDL, \alphaUL)
    &= D_{\mathrm{fib}}^{DL}(\alphaDL) + D_{\mathrm{fib}}^{UL}(\alphaUL),
    \label{eq:D_fib}
\end{align}
and the coupling penalty is
\begin{multline}
    C(\alphaDL, \alphaUL)
    \;=\;
    \bigl[\max\!\bigl(0,\,
        D_{\mathrm{ntn}}(\alphaDL, \alphaUL) - \Cntn\bigr)\bigr]^2
    \;+\; \\
    \bigl[\max\!\bigl(0,\,
        D_{\mathrm{fib}}(\alphaDL, \alphaUL) - \Cfib\bigr)\bigr]^2.
    \label{eq:coupling_penalty}
\end{multline}
Let $\mu > 0$ be the coupling coefficient governing the aggressiveness of
oversubscription penalization. The payoff functions of the two agents are then
\begin{align}
    \JDL(\alphaDL, \alphaUL)
    &= \sum_{u \in \UDL} \pi_u\, U_u(\alpha_u^{\mathrm{DL}})
       - \mu\, C(\alphaDL, \alphaUL),
    \label{eq:JDL} \\[4pt]
    \JUL(\alphaDL, \alphaUL)
    &= \sum_{v \in \UUL} \pi_v\, U_v(\alpha_v^{\mathrm{UL}})
       - \mu\, C(\alphaDL, \alphaUL).
    \label{eq:JUL}
\end{align}
To achieve a Nash equilibrium, each agent repeatedly solves
its best-response problem while the other agent fixes its strategy:
\begin{align}
    \alphaDL{}^*
    &= \arg\max_{\alphaDL \in [0,1]^N}
       \JDL(\alphaDL, \alphaUL{}^*),
    \label{eq:BR_DL}\\[2pt]
    \alphaUL{}^*
    &= \arg\max_{\alphaUL \in [0,1]^M}
       \JUL(\alphaDL{}^*, \alphaUL).
    \label{eq:BR_UL}
\end{align}

\begin{remark}[Decentralized execution]

Since backhaul traffic is bidirectional, each agent's RX counters on the NTN and fiber
interfaces directly measure the opposing agent's aggregate load. Each agent therefore
solves~\eqref{eq:BR_DL}--\eqref{eq:BR_UL} using only local telemetry, without
inter-agent communication.
\end{remark}

\begin{table}[t]
\centering
\caption{Component terms of the per-slice utility $U_u(\alpha_u)$.}
\label{tab:utility_terms}
\renewcommand{\arraystretch}{1.8} 
\begin{tabular}{| >{\raggedright\arraybackslash}m{0.9cm} | >{\raggedright\arraybackslash}m{3.9cm} | >{\centering\arraybackslash}m{2.7cm} |}
\hline
\textbf{Term} & \textbf{Formula} & \textbf{Interpretation} \\
\hline \hline
$T_u(\alpha_u)$ & $\ln\!\bigl(1 + \alpha_u\Cntn + (1-\alpha_u)\Cfib\bigr)$ & Throughput utility (log, proportional fairness) \\
\hline
$L_u(\alpha_u)$ & $-\dfrac{\alpha_u d^{\mathrm{ntn}} + (1-\alpha_u)d^{\mathrm{fib}}}{d^{\max}_{c_u}}$ & Latency utility (normalised, negative cost) \\
\hline
$R_u(\alpha_u)$ & $\alpha_u(1-\ell^{\mathrm{ntn}}) + (1-\alpha_u)(1-\ell^{\mathrm{fib}})$ & Reliability (expected packet survival) \\
\hline
$P_u(\alpha_u)$ & $\sum_{k}\!\left(\dfrac{\max(0,\,m_k(\alpha_u)-s_k^{c_u})}{s_k^{c_u}}\right)^{\!2}$ & SLA violation penalty (quadratic barrier) \\
\hline
\end{tabular}
\end{table}

\subsection{Exact Potential Game Formulation}
\label{sec:exact_potential_game}

We now formalize the game as an \emph{exact potential game} by
constructing an explicit potential function and proving the alignment condition of
Monderer and Shapley~\cite{MONDERER1996124}.

\begin{definition}[Nash equilibrium]
\label{def:NE}
A joint strategy profile
$\bfAlpha^* = (\alphaDL{}^*, \alphaUL{}^*) \in [0,1]^N \times [0,1]^M$
is a \emph{Nash equilibrium} of the game $\Gamma$ if and only if
\begin{align}
    \JDL(\alphaDL{}^*, \alphaUL{}^*)
    &\geq \JDL(\alphaDL, \alphaUL{}^*),
    \quad \forall\,\alphaDL \in [0,1]^N,
    \label{eq:NE_DL}\\[4pt]
    \JUL(\alphaDL{}^*, \alphaUL{}^*)
    &\geq \JUL(\alphaDL{}^*, \alphaUL),
    \quad \forall\,\alphaUL \in [0,1]^M.
    \label{eq:NE_UL}
\end{align}
At a Nash equilibrium, neither agent can unilaterally improve its own payoff by
deviating from $\bfAlpha^*$.
\end{definition}

\begin{definition}[Exact Potential Game \cite{MONDERER1996124}]
\label{def:EPG}
A game $\Gamma = \{\mathcal{N}, \{\mathcal{A}_i\}, \{J_i\}\}$ is an
\emph{exact potential game} if there exists a function
$\Pot : \mathcal{A} \to \mathbb{R}$, called an \emph{exact potential function},
such that for every player $i \in \mathcal{N}$ and for every two strategy profiles
that differ only in player $i$'s action:
\begin{equation}
    J_i(\bm{a}_i,\, \bm{a}_{-i}) - J_i(\bm{a}_i',\, \bm{a}_{-i})
    =
    \Pot(\bm{a}_i,\, \bm{a}_{-i}) - \Pot(\bm{a}_i',\, \bm{a}_{-i}),
    \label{eq:EPG_condition}
\end{equation}
where $\bm{a}_{-i}$ is the strategies of all players except player $i$.
\end{definition}

The key implication of an exact potential game is that any unilateral improvement by one
agent increases $\Pot$ by the exact same amount as it increases that agent's own payoff. As a result, best-response dynamics steadily ascend the potential function, ensuring convergence to a pure Nash equilibrium.

\begin{theorem}[Exact Potential Function]
\label{thm:potential_function}
Let $ \bfAlpha = (\alphaDL, \alphaUL)  $ denote the joint allocation vector. Given the payoff functions $\JDL$ and $\JUL$, the load-balancing game $\Gamma$ is an exact potential game with an exact potential function

\begin{equation}
        \Pot(\bfAlpha) = 
        \sum_{u \in \UDL} \pi_u\, U_u(\alpha_u^{\mathrm{DL}})
        \;+\;
        \sum_{v \in \UUL} \pi_v\, U_v(\alpha_v^{\mathrm{UL}})
        -
        \mu\, C(\bfAlpha).
\label{eq:potential_function}
\end{equation}

\end{theorem}

\begin{proof}
It can be verified by direct substitution that for any unilateral deviation of one agent (holding the other fixed), the change in the payoff of that agent equals the change in $\Pot$, satisfying the exact potential condition of~\cite{MONDERER1996124}.
\end{proof}
\subsection{Equilibrium Analysis and Convergence}

\label{sec:convergence}
Having established that the load-balancing game is an exact potential game, we now prove that it admits a \emph{unique} Nash equilibrium by leveraging the strict concavity
of the potential function $\Pot$.

\begin{lemma}[Existence of Nash equilibrium]
\label{lem:NE_existence}
A pure-strategy Nash equilibrium exists for the game $\Gamma$.
\end{lemma}

\begin{proof}
Since $\Gamma$ is an exact potential game and the strategy space
$\mathcal{A} = [0,1]^{N+M}$ is compact, a maximizer of $\Pot$ over $\mathcal{A}$ exists
by the Extreme Value Theorem. Every maximizer of an exact potential function is a Nash
equilibrium \cite{MONDERER1996124}.
\end{proof}

\begin{lemma}[Uniqueness of Nash equilibrium]
\label{lem:NE_uniqueness}
The game $\Gamma$ has a unique Nash equilibrium
$\bfAlpha^* \in [0,1]^{N+M}$.
\end{lemma}

\begin{proof}
We verify both conditions as included in \cite{MONDERER1996124}:

\noindent\textbf{Condition (i): Compact and convex strategy space.}
Each per-slice allocation $\alpha_u \in [0,1]$ is confined to a closed, bounded interval. Thus, $\mathcal{A} = [0,1]^N \times [0,1]^M = [0,1]^{N+M}$ is a compact convex subset of $\mathbb{R}^{N+M}$.

\noindent\textbf{Condition (ii): Continuous differentiability and
strict concavity of $\Pot$.}
The potential $\Pot$ is continuously differentiable on the interior $(0,1)^{N+M}$ as every component ($T_u$ is a log of an affine function, $L_u$ and $R_u$ are affine, $P_u$ and the coupling term $C(\bf{\alpha})$ are sums of squared-hinge functions) is $C^1$.

For strict concavity, we compute the Hessian $H_{\Pot}(\bfAlpha) \in \mathbb{R}^{(N+M)\times(N+M)}$ and show
it is negative definite. Since each per-slice utility $U_u(\alpha_u^{\mathrm{DL}})$ depends
only on the scalar $\alpha_u^{\mathrm{DL}}$ and analogously for the UL, the Hessian of
the aggregate utility terms is \emph{diagonal}. 
Accordingly,
\begin{equation}
    H_{\Pot}(\bfAlpha)
    \;=\;
    \underbrace{D(\bfAlpha)}_{\text{diagonal}}
    \;-\;
    \mu\,\kappa(\bfAlpha)\;\bm{b}^+(\bm{b}^+)^\top,
    \label{eq:Hessian_Phi}
\end{equation}
where $\bm{b}^+ = (\hat{b}_1^{\mathrm{DL}},\ldots,\hat{b}_N^{\mathrm{DL}},
\hat{b}_1^{\mathrm{UL}},\ldots,\hat{b}_M^{\mathrm{UL}})^\top > \bm{0}$ is the vector of
estimated peak demands, and $\kappa(\bfAlpha) \in \{0,2,4\}$ is a non-negative scalar
arising from the activation of the oversubscription indicators:
\begin{equation}
    \kappa(\bfAlpha)
    \;=\;
    2(
        \mathbf{1}\!\bigl[D_{\mathrm{ntn}} > \Cntn\bigr]
        + \mathbf{1}\!\bigl[D_{\mathrm{fib}} > \Cfib\bigr]
    ).
\end{equation}
The diagonal matrix $D(\bfAlpha)$ has entries, for each slice $u$,
\begin{multline}
    D_{uu}(\bfAlpha)
    \;=\;
    -\pi_u
    \left[
        \frac{w_T^{c_u}(\Delta\hat{C})^2}{\bigl(1 + \Cfib + \alpha_u\Delta\hat{C}\bigr)^2}
        \;+\; \right. \\ \left.
        w_P^{c_u}\,\lambda^{c_u}
        \sum_{k} \frac{2(\Delta m_k)^2}{(s_k^{c_u})^2}
        \,\mathbf{1}[m_k(\alpha_u) > s_k^{c_u}]
    \right]
    \;\leq\; 0,
\label{eq:diagonal_entries}
\end{multline}
where $\Delta m_k = \mathtt{metric}_{k,\mathrm{ntn}} - \mathtt{metric}_{k,\mathrm{fib}}$
and the inequality holds because $w_T^{c_u}, w_P^{c_u}, \lambda^{c_u}, \pi_u \geq 0$.
Under the standing assumption that $\Cntn \neq \Cfib$ (NTN and fiber have different
capacities), we have $\Delta\hat{C} \neq 0$, and therefore
\begin{multline}
    D_{uu}(\bfAlpha)
    \;\leq\;
    -\pi_u \cdot
    \frac{w_T^{c_u}(\Delta\hat{C})^2}{\bigl(1 + \Cfib + \alpha_u\Delta\hat{C}\bigr)^2}
    \;<\; 0.
    \label{eq:diagonal_strict_neg}
\end{multline}
Hence $D(\bfAlpha) \prec 0$ (strictly negative definite diagonal matrix).
Furthermore, $-\mu\kappa(\bfAlpha)\,\bm{b}^+(\bm{b}^+)^\top$ is a negative
semi-definite rank-1 correction. Therefore $H_{\Pot}(\bfAlpha) \prec 0$ for all $\bfAlpha \in (0,1)^{N+M}$, and since
$\Pot$ is continuous on the compact set $[0,1]^{N+M}$, it follows that $\Pot$ is
\emph{strictly concave} on $[0,1]^{N+M}$.

Since both conditions hold, the Nash
equilibrium of the load-balancing game $\Gamma$ is \textbf{unique}.
\end{proof}

\section{Experimental Setup}
The experimental evaluation is conducted on a geographically distributed 5G testbed, with the RAN deployed in Malaga and the core deployed in Ghent. This setup emulates a hybrid backhaul scenario in which TN and NTN links coexist and are jointly used for traffic delivery. As depicted in Figure~\ref{fig:archi}, 
a set of UE devices connect to Nokia Airscale gNB, which forwards uplink traffic to an edge-based load-balancing virtual machine acting as the uplink routing agent. This edge agent connects to a corresponding downlink load balancer at the core via two parallel virtual tunnels: a low-latency terrestrial fiber link and an NTN satellite link utilizing Starlink.

\begin{table}[b]
\centering
\caption{SLA constraints per traffic class.}
\label{tab:sla}
\renewcommand{\arraystretch}{1.15}
\setlength{\tabcolsep}{4pt}
\footnotesize
\begin{tabular}{@{}l c c c c@{}}
\toprule
\textbf{Traffic Class} & \textbf{Priority} & \textbf{Max RTT} & \textbf{Max Jitter} & \textbf{Max Loss} \\
 &  & \textbf{(ms)} & \textbf{(ms)} & \textbf{(\%)} \\
\midrule
V2X & High   & 60  & 15  & 0.5 \\
Emergency & High & 70 & 20 & 0.5 \\
Video Streaming & Medium    & 200 & 80 & 3.0 \\
IoT & Low    & 500 & 150 & 10.0 \\
Best-effort & Low    & 800 & 100 & 5.0 \\

\bottomrule
\end{tabular}
\end{table}

To evaluate the framework defined in Section \ref{sec:model}, bidirectional traffic is generated to emulate five distinct slices: V2X, emergency services, video streaming, IoT, and best-effort web traffic. We configured the testbed with the specific transport-layer SLA constraints detailed in Table \ref{tab:sla}. These values are configured to reflect the geographic distance and the physical characteristics of the Starlink NTN environment. Generated flows representing the five slices are explicitly non-stationary, incorporating randomized burst durations, pause intervals, and variable bandwidth demands.

The real-time network state is captured by a custom metrics exporter running natively on the load balancers that exposes continuous telemetry such as interface-level throughput, RTT, jitter, and packet loss to a time-series database. Routing decisions are governed by two separate controllers deployed at the uplink and downlink load balancers. These controllers utilize the Best Response Iteration (BRI) algorithm to solve the proposed potential game and determine the optimal traffic allocation vectors for each slice. Since the potential function $\Pot$ is strictly concave, BRI is guaranteed to converge to the unique pure Nash equilibrium.

The proposed potential game approach is benchmarked in the testbed against a set of representative baseline routing strategies commonly used in hybrid and slice-aware traffic steering, including equal split, weighted round robin, random allocation, and a purpose-built SLA-aware heuristic. The heuristic baseline computes a per-slice score for each link by combining normalized RTT, loss, and jitter headroom with traffic-class weights on both links and performs redistribution sweeps that shift load from any link exceeding 90\% utilization to the alternate link. 
The main evaluation metrics are aggregate network metrics and per-slice SLA violation rates.

\section{Results and Discussion}
Table \ref{tab:metrics} summarizes the average network metrics recorded across the 300-second non-stationary traffic episodes. Here, the RTT represents the effective RTT proportionally weighted by the actual volume of traffic steered over the TN and NTN links. Similarly, the reported throughput reflects the total bidirectional traffic achieved by the system under the same traffic scenario.

The table shows that the potential game framework consistently outperforms the other algorithms in maximizing overall network efficiency. It achieved the lowest effective RTT (51.683 ms), the highest bidirectional throughput (78.195 Mbps), and the lowest packet loss rate of 0.2\%. Although these aggregate improvements are modest, the primary benefit lies in the substantial reduction in SLA violations per slice, as shown separately in the heatmap of Figure 2. Nonetheless, this marginal improvement highlights the benefits of the utility evaluation to proactively avoid congestion. Fairness is calculated using Jain's fairness index \cite{jainQuantitativeMeasureFairness1998}, applied to per-link throughput utilization across the two backhaul paths. 

\begin{figure}[b]
    \centering
    \includegraphics[width=0.90\linewidth]{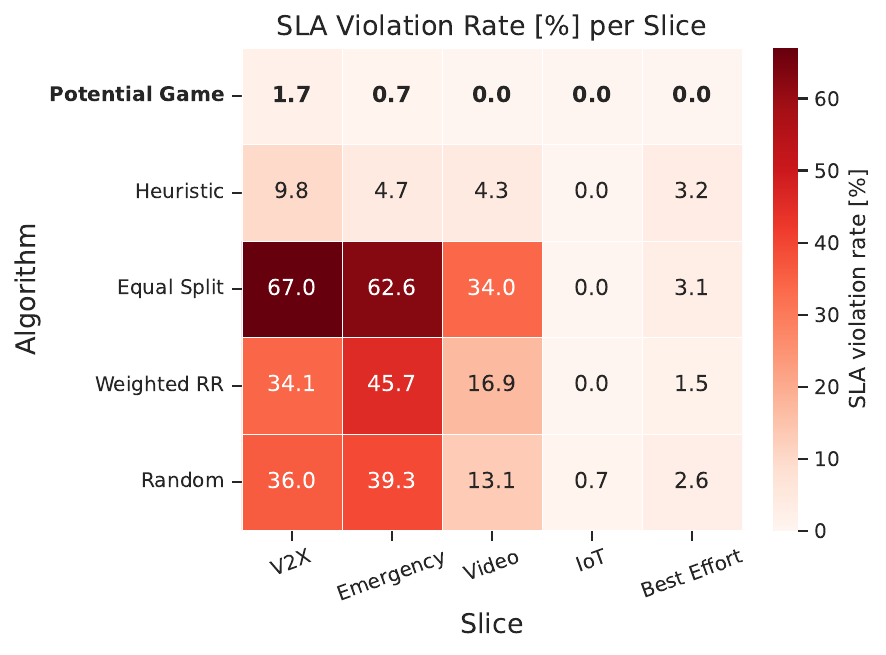}
    \caption{Heatmap of the average SLA violation rates per traffic slice.}
    \label{fig:heatmap}
\end{figure}

The Equal Split baseline naturally achieved a perfect fairness score of 1.000, as it indiscriminately distributes traffic. However, enforcing strict arithmetic fairness across highly asymmetric links increases packet loss and RTT. By forcing 50\% of latency-sensitive and high-bandwidth traffic onto the constrained satellite link, Equal Split yields the highest packet loss (1.445\%) among all methods. In contrast, the potential game achieves a fairness score of 0.963 while simultaneously delivering the best throughput, latency, and reliability, proving that the algorithm successfully balances individual slice requirements without starving any specific flow.

\begin{table}[t]
\centering
\caption{Average network-level performance metrics per algorithm.}
\label{tab:metrics}
\renewcommand{\arraystretch}{1.15}
\setlength{\tabcolsep}{4pt}
\footnotesize
\begin{tabular}{@{}l c c c c@{}}
\toprule
\textbf{Algorithm} & \textbf{RTT} & \textbf{Loss} & \textbf{Throughput}  & \textbf{Fairness} \\
 & \textbf{(ms)} & \textbf{(\%)} & \textbf{(Mbps)} & \\
\midrule
Potential Game & \textbf{51.683} & \textbf{0.200} & \textbf{78.195} & 0.963 \\
Heuristic &	53.884	& 0.349 & 75.307  & 0.855 \\
Equal Split	& 54.801 & 1.445 & 75.193 & \textbf{1.000} \\
Weighted RR & 53.463 & 0.901 & 74.374 & 0.709 \\
Random	& 54.793 & 1.164 & 70.002 & 0.732\\
\bottomrule
\end{tabular}
\end{table}

While aggregate metrics provide a macro-level view of system efficiency, the controller's primary objective is to maintain strict SLA compliance across heterogeneous slices. Figure \ref{fig:heatmap} displays a heatmap of the SLA violation for each load-balancing algorithm over the evaluation period. 

The potential game controller demonstrates superior performance in safeguarding strict SLA requirements. For less-sensitive slices, such as IoT, video, and best effort, it completely eliminates  SLA violations. For sensitive slices, it reduces the SLA violations to 1.7\% for V2X and 0.7\% for the emergency slice. These residual violations occur only during severe congestion spikes, when even the optimal allocation cannot fully prevent brief queuing delays on the terrestrial link for slices with tight SLA requirements. 


The IoT slice remains highly robust across almost all configurations. This indicates that its loose SLA requirements are easily satisfied across the hybrid backhaul. The primary challenge in backhaul load-balancing is intelligently steering ultra-low-latency V2X and emergency flows alongside high-bandwidth video and best-effort flows to avoid the NTN latency bottleneck and terrestrial link congestion. 

Finally, the temporal evolution of the link-specific metrics and the objective function is depicted in Figure \ref{fig:timeseries}. Throughout the evaluation period, the potential function $\Pot$ exhibits a stable baseline convergence plateau, punctuated by transient, distinct drops. These drops correlate with mathematical penalties incurred when SLA constraints are breached, primarily driven by RTT degradation (second panel) and packet loss events (third panel)  during severe congestion spikes.

 \begin{figure}[t]
    \centering
    \includegraphics[width=0.98\linewidth]{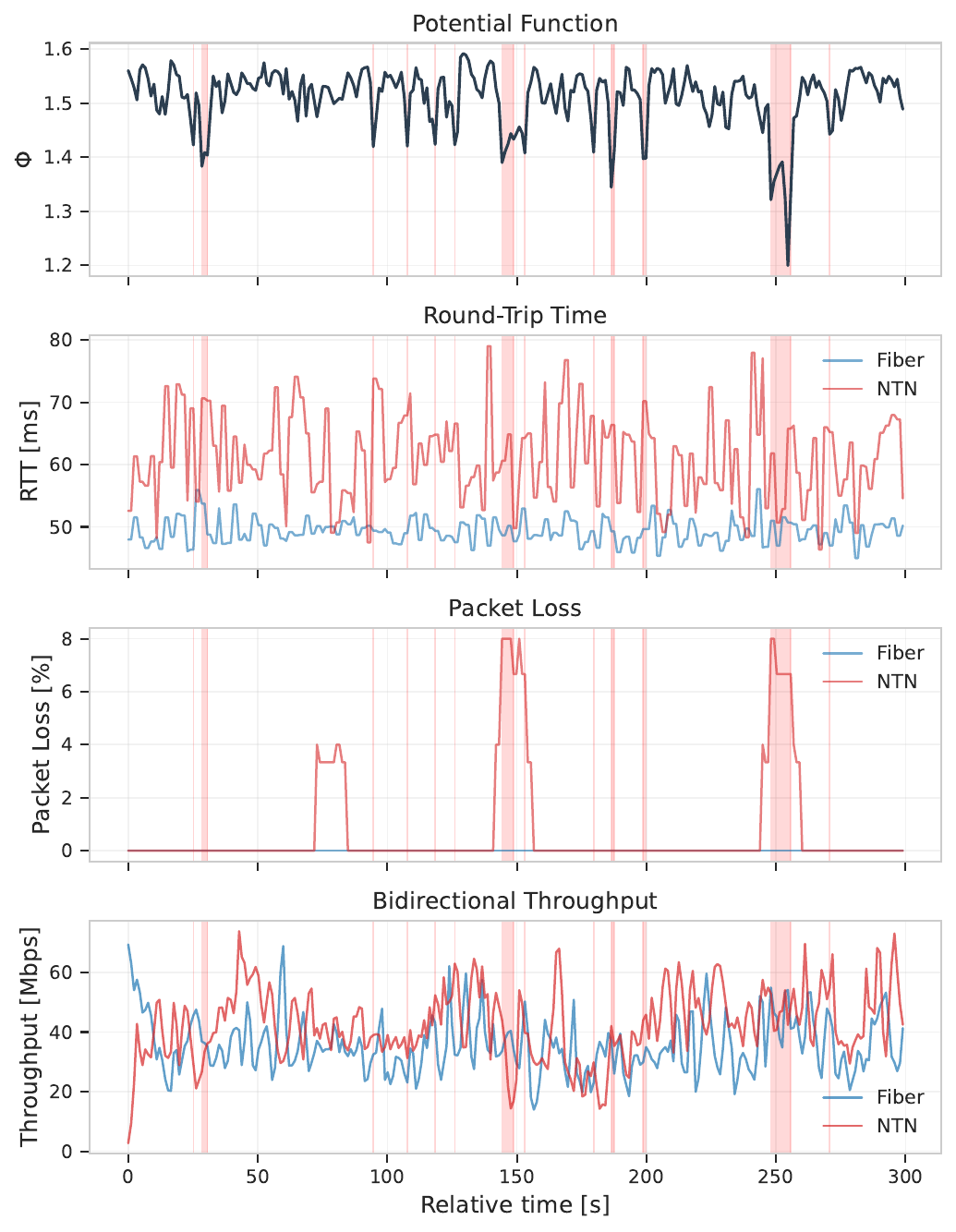}
    \caption{Temporal evolution of the potential function and network metrics.}
    \label{fig:timeseries}
\end{figure}

It is critical to note that high latency does not universally penalize the system. Drops in $\Pot$ are triggered only when physical network impairments violate an assigned SLA, such as when congestion forces time-sensitive V2X or emergency slices to experience queuing delays or packet drops. 

Furthermore, Figure \ref{fig:timeseries} highlights the reactive agility of the potential game framework. Upon detecting a sub-optimal utility state (a drop in $\Pot$), the decentralized controller immediately executes new best-response sweeps and shifts the fractional allocation to alleviate the bottleneck. Consequently, the potential function sharply recovers to its maximized equilibrium state. This pattern confirms that the controller continuously tracks a shifting pure Nash equilibrium in real time.

The proposed controller exhibits low computational overhead, making it suitable for real-time deployment. Each agent solves a continuous optimization problem over its per-slice allocation variables, with linear complexity  $\mathcal{O}(N)$ or $\mathcal{O}(M)$ per iteration.
In our testbed deployment, each best-response computation completes in an average of 4.412~ms on a standard desktop CPU (Intel Core i9-9900K). 
The alternating best-response dynamics typically reach equilibrium in less than five iterations under typical load conditions.
\section{Conclusion}

Integrating hybrid TN-NTN backhauls enhances 5G/6G resilience but complicates SLA management. To address this, we introduced a decentralized load-balancing controller that models per-slice traffic steering as an exact potential game. Testbed evaluations demonstrated that our approach significantly outperforms heuristic baselines under non-stationary loads, ensuring strict SLA compliance. Future work will adapt this framework for moving LEO constellation dynamics and integrate it with 5G core orchestration for autonomous, cross-layer SLA management.

\bibliographystyle{myIEEEtran}
\bibliography{references}

\end{document}